\documentclass[a4paper,11pt]{article}
\usepackage[top=30truemm,bottom=30truemm,left=25truemm,right=25truemm]{geometry}

\usepackage{amsthm,mathtools,bm,color}
\usepackage{tikz}
\usetikzlibrary{matrix}

\usepackage{graphicx}
\usepackage[small]{caption}

\usepackage{hyperref}

\theoremstyle{plain}% Theorem-like structures provided by amsthm.sty
\newtheorem{theorem}{Theorem}[section]
\newtheorem{lemma}[theorem]{Lemma}

\newtheorem{proposition}[theorem]{Proposition}

\begin{document}

\title{Application of the extended $q$-discrete Toda equation to computing eigenvalues of Hessenberg totally nonnegative matrices}
\author{R. Watanabe\footnote{Graduate School of Informatics, Kyoto University (watanabe.ryoto.37p@st.kyoto-u.ac.jp)}, \and M. Shinjo\footnote{Faculty of Science and Engineering, Doshisha University (mshinjo@mail.doshisha.ac.jp)}, \and Y. Yamamoto\footnote{Department of Communication Engineering and Informatics, The University of Electro-Communications (yusaku.yamamoto@uec.ac.jp)} \and and \and M. Iwasaki\footnote{Faculty of Life and Environmental Sciences, Kyoto Prefectural University (imasa@kpu.ac.jp)}}
\maketitle
%
% ---------------------------------------------------------------------
%Abstract
\begin{abstract}
The Toda equation is one of the most famous integrable systems, and its time-discretization is simply the recursion formula of the quotient-difference (qd) algorithm for computing eigenvalues of tridiagonal matrices. An extension of the Toda equation is the $q$-Toda equation, which is derived by replacing standard derivatives with the so-called $q$-derivatives involving a parameter $q$ such that $0<q<1$. In our previous paper, we showed that a discretization of the $q$-Toda equation is shown to be also applicable to computing tridiagonal eigenvalues. In this paper, we consider another extension of the $q$-discrete Toda equation and find an application to computing eigenvalues of Hessenberg totally nonnegative (TN) matrices, which are matrices where all minors are nonnegative. There are two key components to our approach. First, we consider the extended $q$-discrete equation from the perspective of shifted $LR$ transformations, similarly to the discrete Toda and its $q$-analogue cases. Second, we clarify asymptotic convergence as discrete-time goes to infinity in the $q$-discrete Toda equation by focusing on TN properties. We also present two examples to numerically verify convergence to Hessenberg TN eigenvalues numerically. 
\end{abstract}
%
%\begin{keywords}$q$-derivative; Toda equation; integrable system; eigenvalue; totally nonnegative matrix\end{keywords}
%
%\begin{amscode}``37N30", ``65F15"\end{amscode}
%
%
% ---------------------------------------------------------------------
\section{Introduction}
\label{sec:1}
% ---------------------------------------------------------------------
%
Some integrable systems have interesting relationships to algorithms for computing matrix eigenvalues. The oldest observation is the case of the Toda equation \cite{Toda_1967} which describes mass motion governed by nonlinear springs. The time evolution from $t=0$ to $t=1$ in the Toda equation corresponds to $1$-step of the $QR$ algorithm, which generates similarity transformations of exponential of a tridiagonal matrix \cite{Symes_1982}. 
A discrete-time version of the Toda equation \cite{Hirota_1981} is simply the recursion formula of the quotient-difference (qd) algorithm \cite{Rutishauser_1990} for computing eigenvalues of tridiagonal matrices. The qd algorithm generates a series of $LR$ transformations which decompose tridiagonal matrices into products of lower and upper bidiagonal matrices and then reverses the order of the products. Thus, discrete-time evolutions in the discrete Toda equation also perform this function. An integrable system closely related to the Toda equation is the Lotka--Volterra (LV) system, which describes the simplest prey-predator model. The LV system has asymptotic convergence as time goes to infinity to singular values of a bidiagonal matrix, or equivalently, to eigenvalues of a positive-definite tridiagonal matrix \cite{Chu_1986}. The discrete LV (dLV) system thus enables us to design a numerical algorithm for bidiagonal singular values \cite{Iwasaki_2002,Iwasaki_2004}. Similarly to the discrete Toda case, discrete-time evolutions in the dLV system can be regarded as generating tridiagonal $LR$ transformations.
\par
Discrete hungry LV (dhLV) systems \cite{Tsujimoto_1993,Suris_1996} are an extension of the dLV system in which each species is assumed to prey on one, two, or more species. See \cite{Bogoyavlenskii_1991} and \cite{Itoh_1987} for continuous-time versions of the dhLV systems. The discrete hungry Toda (dhToda) equations \cite{Tokihiro_1999,Fukuda_2013_1}, which link to the dhLV systems, are of course extensions of the discrete Toda equation. These discrete hungry integrable systems are applicable to computing eigenvalues of Hessenberg totally nonnegative (TN) matrices, which are matrices where all minors are nonnegative \cite{Fukuda_2013_1}. However, because they employ similarity transformations based on repeating the tridiagonal $LR$ transformations, they require the target matrix to be expressed as a product of bidiagonal matrices in the initial setting. 
\par
The Kostant Toda equation~\cite{Kostant_1979} is an extension of the continuous Toda equation. It is related to the discrete Korteweg de Vries equation~\cite{Rolania_2019}, which is equivalent to the continuous hungry LV system. The continuous analogue of the dhToda equation is a special case of the Kostant Toda equation~\cite{Shinjo_2020}. In this paper, we focus on the discrete systems.
\par
The $q$-analogue of the Toda equation \cite{Area_2018}, which involves a parameter $q$ satisfying $0<q<1$, is another extension of the Toda equation. Our previous paper \cite{Watanabe_2022} associated the $q$-Toda equation with the tridiagonal eigenvalue problem, and showed that its time-discretization generates tridiagonal similarity transformations which do not require bidiagonal factorizations. In this paper, we consider a further extension of the $q$-discrete Toda equation and relate it to similarity transformations of Hessenberg matrices. Moreover, by examining asymptotic convergence as discrete time goes to infinity in the extended $q$-discrete Toda equation, we show that it can compute the eigenvalues of TN Hessenberg matrices without employing bidiagonal factorizations.
\par
The remainder of this paper is organized as follows. In Section \ref{sec:2}, we derive an extension of the $q$-discrete Toda equation which is related to similarity transformations of Hessenberg matrices. We also find a continuous analogue of the extended $q$-discrete Toda equation. In Section \ref{sec:3}, we introduce the implicit $L$ theorem to consider the similarity transformations from the perspective of shifted $LR$ transformations whose targets are TN Hessenberg matrices rather than tridiagonal matrices. In Section \ref{sec:4}, by imposing the TN structure on the similarity transformations, we clarify asymptotic convergence as discrete time goes to infinity to TN Hessenberg eigenvalues. We also present numerical examples to demonstrate the asymptotic convergence. Finally, in Section \ref{sec:5}, we give some concluding remarks.
%
%
% ---------------------------------------------------------------------
\section{Extended $q$-discrete Toda equation}
\label{sec:2}
% ---------------------------------------------------------------------
%
In this section, we describe an extension of the $q$-discrete of Toda equation, and then clarify its relationship to similarity transformations of a Hessenberg matrix. 
\par
From Area et al. \cite{Area_2018}, a $q$-analogue of the Toda equation with a parameter $0<q<1$ is defined by:
\begin{align}
\label{eqn:qToda}
\left\{ 
\begin{aligned}
& \dfrac{x_{i}(t;q)-x_{i}(qt;q)}{(1-q)t} = g_{i}(qt;q) - g_{i-1}(qt;q), \quad i=1,2,\ldots,m, \\ 
& \dfrac{y_{i}(t;q)-y_{i}(qt;q)}{(1-q)t}=g_{i}(qt;q) \left( x_{i+1}(qt;q) - x_{i}(t;q) \right), \quad i=1,2,\ldots,m-1, \\
& y_{0}(t;q) \coloneqq 0, \quad y_{m}(t;q) \coloneqq 0, 
\end{aligned}
\right.
\end{align}
where the auxiliary variables $g_{i}(qt;q)$ are given as:
\begin{align}
\label{eqn:qToda_g}
\left\{ 
\begin{aligned}
& g_{1}(qt;q) = \dfrac{ y_{1}(qt;q) }{ 1+(1-q)tx_{1}(qt;q) }, \quad 
g_{i}(qt;q) = \dfrac{ y_{i}(qt;q) }{ y_{i-1}(t;q) } g_{i-1}(qt;q), \quad i=2,3,\ldots,m-1, \\
& g_{0}(qt;q) \coloneqq 0, \quad g_{m}(qt;q) \coloneqq 0. 
\end{aligned}
\right.
\end{align}
As $q\to 1$, we see that $(x_{i}(t;q)-x_{i}(qt;q))/[(1-q)t]\to dx_{i}(t)/dt$, $(y_{i}(t;q)-y_{i}(qt;q))/[(1-q)t]\to dy_{i}(t)/dt$ and $g_{i}(t;q)\to y_{i}(t)$. Thus, we can easily check that the Toda equation is derived by taking the limit $q\to 1$ in the $q$-Toda equation \eqref{eqn:qToda}. 
We introduce a  time variable $t^{(n)}\coloneqq q^{-n}t^{(0)}$, and replace $x_i(t^{(n)};q),y_i(t^{(n)};q)$, and $g_i(t^{(n)};q)$ with $x_i^{(n)},y_i^{(n)}$, and $g_i^{(n)}$, respectively. Then, we can rewrite the $q$-Toda equation \eqref{eqn:qToda} as:
\begin{align}
\label{eqn:qdToda}
\left\{ 
\begin{aligned}
& x_{i}^{(n+1)}=x_{i}^{(n)}+(t^{(n+1)}-t^{(n)})\left(g_{i}^{(n)}-g_{i-1}^{(n)}\right), \quad i=1,2,\ldots,m, \\ 
& y_{i}^{(n+1)}=y_{i}^{(n)}+(t^{(n+1)}-t^{(n)})g_{i}^{(n)}\left( x_{i+1}^{(n)}-x_{i}^{(n+1)} \right), \quad i=1,2,\ldots,m-1, \\
& y_{0}^{(n)} \coloneqq 0, \quad y_{m}^{(n)} \coloneqq 0, 
\end{aligned}
\right.
\end{align}
where the auxiliary variables $g_{i}^{(n)}$ satisfy:
\begin{align*}
\left\{ 
\begin{aligned}
& g_{1}^{(n)} \coloneqq \dfrac{ y_{1}^{(n)} }{ 1+(t^{(n+1)}-t^{(n)})x_{1}^{(n)} }, \quad 
g_{i} = \dfrac{ y_{i}^{(n)} }{ y_{i-1}^{(n+1)} } g_{i-1}^{(n)}, \quad i=2,3,\ldots,m-1, \\
& g_{0}^{(n)} \coloneqq 0, \quad g_{m}^{(n)} \coloneqq 0. 
\end{aligned}
\right.
\end{align*}
We refer to \eqref{eqn:qdToda} as the $q$-discrete Toda equation. Our previous paper \cite{Watanabe_2022} related it to similarity transformations of tridiagonal matrices. We also showed that an extension of the $q$-discrete Toda equation can be applied to computing eigenvalues of $M$-tridiagonal matrices with two nonzero off-diagonals consisting of the $(1,M+1),(2,M+2),\dots,(N-M-1,N)$ and $(M+1,1),(M+2,2),\dots,(N,N-M-1)$ entries. 
\par
In this paper, we consider another extension of the $q$-discrete Toda equation and relate it to similarity transformations of a Hessenberg matrix. Our new extension has parameters $M$ and $\mu^{(n)}$ and can be written as follows:
\begin{align}
\label{eqn:extended_qdToda}
\left\{ 
\begin{aligned}
& x_{i,j}^{(n+1)} = x_{i,j}^{(n)} 
+ \mu^{(n)} \left( x_{i,j+1}^{(n)} g_{j}^{(n)} - g_{i-1}^{(n)} x_{i-1,j}^{(n+1)} \right), 
\quad i=1,2,\ldots,m, \quad j=i,i+1,\ldots,\ell_{i-1}, \\ 
& x_{i,\ell_{i}}^{(n+1)} = x_{i,\ell_{i}}^{(n)} 
+ \mu^{(n)} \left( g_{\ell_{i}}^{(n)} - g_{i-1}^{(n)} \right), 
\quad i=1,2,\ldots,m-M+1, \\ 
& y_{i}^{(n+1)} = y_{i}^{(n)} 
+ \mu^{(n)} g_{i}^{(n)} \left( x_{i+1,i+1}^{(n)} - x_{i,i}^{(n+1)} \right), 
\quad i=1,2,\ldots,m-1, \\
& x_{0,1}^{(n)} \coloneqq 0, \quad x_{0,2}^{(n)} \coloneqq 0, \quad \ldots, \quad x_{0,M-1}^{(n)} \coloneqq 0, \\
& x_{m-M+2,m+1}^{(n)} \coloneqq 0, \quad x_{m-M+3,m+1}^{(n)} \coloneqq 0, 
\quad \ldots, \quad x_{m,m+1}^{(n)} \coloneqq 0, \\
& y_{0}^{(n)} \coloneqq 0, \quad y_{m}^{(n)} \coloneqq 0, 
\end{aligned}
\right.
\end{align}
where $\ell_{i}\coloneqq\min(i+M-1,m)$, and $g_{i}^{(n)}$ is an auxiliary variable given by:
\begin{align}
\label{eqn:extended_qdToda_g}
\left\{ 
\begin{aligned}
& g_{1}^{(n)} \coloneqq \dfrac{ y_{1}^{(n)} }{ 1 + \mu^{(n)} x_{1,1}^{(n)} }, \quad
g_{i} = \dfrac{ y_{i}^{(n)} }{ y_{i-1}^{(n+1)} } g_{i-1}^{(n)}, \quad i=2,3,\ldots,m-1, \\
& g_{0}^{(n)} \coloneqq 0, \quad g_{m}^{(n)} \coloneqq 0. 
\end{aligned}
\right.
\end{align}
As will be shown below, this extension corresponds to extending the tridiagonal matrix associated with the $q$-discrete Toda equation to a Hessenberg matrix. Note that the extended $q$-discrete Toda equation \eqref{eqn:extended_qdToda} with $M=1$ coincides with the $q$-discrete Toda equation \eqref{eqn:qdToda}. We can determine the sequences $\{ x_{i,j}^{(n)} \}_{n=0,1,\ldots}$, $\{ y_{i}^{(n)} \}_{n=0,1,\ldots}$, and $\{ g_{i}^{(n)} \}_{n=0,1,\ldots}$ uniquely if $x_{i,j}^{(0)}$, $y_{i}^{(0)}$, and $g_{i}^{(0)}$ and $\mu^{(n)}$ are given. 
See also Figure \ref{fig:diagram_time_evolution} illustrates the time evolution from $n$ to $n+1$ in the extended $q$-discrete Toda equation \eqref{eqn:extended_qdToda}. 
% ---------------------------------------------------------------------
\begin{figure}[tb]
\begin{center}
\[
\begin{tikzpicture}[
every node/.style={anchor=base,align=center,text depth=1ex,text height=2ex,text width=2.5em},
]
\matrix(A)[
matrix of math nodes,nodes in empty cells
] {
A^{(n)} & & & & G^{(n)} & & A^{(n+1)} & &
\\
\cdots & x_{1,j}^{(n)} & x_{1,j+1}^{(n)} & \cdots & & g_{1}^{(n)} & \cdots & x_{1,j}^{(n+1)} & \cdots
\\
& \vdots & \vdots & & & \vdots & & \vdots &
\\
\cdots & x_{i-1,j}^{(n)} & x_{i-1,j+1}^{(n)} & \cdots & & g_{i-1}^{(n)} & \cdots & x_{i-1,j}^{(n+1)} & \cdots 
\\
& x_{i,j}^{(n)} & x_{i,j+1}^{(n)} & & & g_{i}^{(n)} & & [entry-box] x_{i,j}^{(n+1)} & 
\\
& \vdots & \vdots & & & \vdots & & \vdots & 
\\
\ddots & x_{j,j}^{(n)} & x_{j,j+1}^{(n)} & \cdots & & g_{j}^{(n)} & \ddots & x_{j,j}^{(n+1)} & \cdots
\\
& y_{j}^{(n)} & x_{j+1,j+1}^{(n)} & & & g_{j+1}^{(n)} & & y_{j}^{(n+1)} & 
\\
 & & y_{j+1}^{(n)} & \ddots & & \vdots & & & \ddots \\
};
% -------------------------------------------------------
%
\draw[dashed](A-1-1.south west)--(A-1-9.south east);
\draw[rounded corners=2pt](A-1-1.north west)--(A-9-1.south west)--(A-9-9.south east)--(A-1-9.north east)--cycle;
\draw[](A-1-5.north west)--(A-9-5.south west);
\draw[](A-1-7.north west)--(A-9-7.south west);
% -------------------------------------------------------
%
\fill[fill=lightgray,opacity=0.3](A-5-2.north west)--(A-5-3.north east)--(A-5-3.south east)--(A-5-2.south west)--cycle;
\fill[fill=lightgray,opacity=0.3](A-8-2.north west)--(A-8-3.north east)--(A-8-3.south east)--(A-8-2.south west)--cycle;
\fill[fill=lightgray,opacity=0.3](A-4-6.north west)--(A-4-6.south west)--(A-4-6.south east)--(A-4-6.north east)--cycle;
\fill[fill=lightgray,opacity=0.3](A-7-6.north west)--(A-7-6.south west)--(A-7-6.south east)--(A-7-6.north east)--cycle;
\fill[fill=lightgray,opacity=0.3](A-4-8.north west)--(A-4-8.south west)--(A-4-8.south east)--(A-4-8.north east)--cycle;
\fill[fill=lightgray,opacity=0.3](A-7-8.north west)--(A-7-8.south west)--(A-7-8.south east)--(A-7-8.north east)--cycle;
\fill[fill=lightgray,opacity=0.3](A-9-3.north west)--(A-9-3.south west)--(A-9-3.south east)--(A-9-3.north east)--cycle;
\draw[thick] (A-5-8.north west)--(A-5-8.south west)--(A-5-8.south east)--(A-5-8.north east)--cycle;
\draw[thick] (A-8-8.north west)--(A-8-8.south west)--(A-8-8.south east)--(A-8-8.north east)--cycle;
\draw[thick] (A-8-6.north west)--(A-8-6.south west)--(A-8-6.south east)--(A-8-6.north east)--cycle;
% -------------------------------------------------------
%
\draw[->](A-5-3) to [out=15,in=165] node {} (A-5-8);
\draw[->](A-8-3) to [out=15,in=165] node {} (A-8-8);
\draw[->](A-7-6) to [out=180,in=180] node {} (A-8-6);
\draw[->](A-4-8) to [out=0,in=0] node {} (A-5-8);
\draw[->](A-7-8) to [out=0,in=0] node {} (A-8-8);
\draw[->](A-4-6) to [out=0,in=180] node {} (A-5-8);
\draw[->](A-7-6) to [out=0,in=180] node {} (A-8-8);
\draw[->](A-7-6) to [out=0,in=180] node {} (A-5-8);
\draw[->](A-8-8) to [out=270,in=270] node {} (A-8-6);
\draw[->](A-9-3) to [out=0,in=180] node {} (A-8-6);
%
% -------------------------------------------------------
\end{tikzpicture}
\]
\caption{Discrete-time evolution in the extended $q$-discrete Toda equation \eqref{eqn:extended_qdToda}. \label{fig:diagram_time_evolution}}
\end{center}
\end{figure}
% ---------------------------------------------------------------------
\par
Now, we introduce $m$-by-$m$ Hessenberg matrices:
\begin{equation}
A^{(n)} \coloneqq
\left( \begin{array}{cccccc}
x_{1,1}^{(n)} & \cdots & x_{1,\ell_{1}}^{(n)} & 1 & & \\
y_{1}^{(n)} & x_{2,2}^{(n)} & \cdots & x_{2,\ell_{2}}^{(n)} & \ddots & \\
 & y_{2}^{(n)} & x_{3,3}^{(n)} & \ddots & \ddots & 1 \\
 & & \ddots & \ddots & \ddots & x_{m-M+1,m}^{(n)} \\
 & & & \ddots & \ddots & \vdots \\
 & & & & y_{m-1}^{(n)} & x_{m,m}^{(n)} \\
\end{array} \right),
\label{eq:An} 
\end{equation}
and $m$-by-$m$ lower bidiagonal matrices:
\[
G^{(n)} \coloneqq
\left( \begin{array}{cccc}
0 & & & \\
g_{1}^{(n)} & 0 & & \\
 & \ddots & \ddots & \\
 & & g_{m-1}^{(n)} & 0
\end{array} \right). 
\]
Then, we can represent the extended $q$-discrete Toda equation \eqref{eqn:extended_qdToda} as a matrix equation. 
\begin{theorem}
\label{thm:similarity_transformation}
A matrix representation of the extended $q$-discrete Toda equation \eqref{eqn:extended_qdToda} is given as:
\begin{equation}
\label{eqn:similarity_transformation}
A^{(n+1)}=A^{(n)}+\mu^{(n)}\left(A^{(n)}G^{(n)}-G^{(n)}A^{(n+1)}\right).
\end{equation}
\end{theorem}
\begin{proof}
Observing the $(i,i), (i,i+1), \ldots, (i,\ell_{i-1})$, and $(i,\ell_{i})$ entries of $A^{(n)}G^{(n)}-G^{(n)}A^{(n+1)}$, we easily derive 
$x_{i,i+1}^{(n)}g_{i}^{(n)}-g_{i-1}^{(n)}x_{i-1,i}^{(n+1)}, x_{i,i+2}^{(n)}g_{i+1}^{(n)}-g_{i-1}^{(n)}x_{i-1,i+1}^{(n+1)}, \ldots,$ 
$x_{i,\ell_{i}}^{(n)}g_{\ell_{i-1}}^{(n)}-g_{i-1}^{(n)} x_{i-1,\ell_{i-1}}^{(n+1)}$, and $g_{\ell_{i}}^{(n)}-g_{i-1}^{(n)}$, respectively. 
The $(i+2,i)$ and $(i+1,i)$ entries of $A^{(n)}G^{(n)}-G^{(n)}A^{(n+1)}$ are, respectively, $y_{i+1}^{(n)}g_{i}^{(n)}-g_{i+1}^{(n)}y_{i}^{(n+1)}$ and $x_{i+1,i+1}^{(n)}g_{i}^{(n)}-g_{i}^{(n)}x_{i,i}^{(n+1)}$. 
The other entries of $A^{(n)}G^{(n)}-G^{(n)}A^{(n+1)}$ are all $0$. 
Moreover, it follows from the second equation of \eqref{eqn:extended_qdToda_g} that $y_{i+1}^{(n)}g_{i}^{(n)}-g_{i+1}^{(n)}y_{i}^{(n+1)}=0$. 
Thus, the extended $q$-discrete Toda equation \eqref{eqn:extended_qdToda} satisfies the matrix representation \eqref{eqn:similarity_transformation}. 
\end{proof}
\par
Since ${\rm det}(I+\mu^{(n)}G^{(n)})=1$, where $I$ denotes the $m$-by-$m$ identity matrix, the inverse matrix $(I+\mu^{(n)}G^{(n)})^{-1}$ always exists. Theorem \ref{thm:similarity_transformation} thus leads to:
\[
A^{(n+1)}=(I+\mu^{(n)}G^{(n)})^{-1}A^{(n)}(I+\mu^{(n)}G^{(n)}), 
\]
which implies that the extended $q$-discrete Toda equation \eqref{eqn:extended_qdToda} generates a similarity transformation of the Hessenberg matrix $A^{(n)}$.
%
% ---------------------------------------------------------------------
\par
We now derive a continuous analogue of the extended $q$-discrete Toda equation \eqref{eqn:extended_qdToda}. 
By setting $x_{i,j}^{(n)}\coloneqq x_i(t^{(n)};q),y_i^{(n)}\coloneqq y_i(t^{(n)};q),g_i^{(n)}\coloneqq g_i(t^{(n)};q)$, and $\mu^{(n)}\coloneqq t^{(n+1)}-t^{(n)}$ and taking the limit as $q\to1$ of the extended $q$-discrete Toda equation \eqref{eqn:extended_qdToda}, we obtain:
\begin{align}
\label{eqn:continuous_extended_qdToda}
\left\{ 
\begin{aligned}
& \dfrac{x_{i,j}(t)}{dt} = x_{i,j+1}(t) y_{j}(t) - y_{i-1}(t) x_{i-1,j}(t),\quad i=1,2,\ldots,m, \quad j=i,i+1,\ldots,\ell_{i-1}, \\ 
& \dfrac{x_{i,\ell_i}(t)}{dt} = y_{\ell_{i}}(t)-y_{i-1}(t),\quad i=1,2,\ldots,m-M+1, \\ 
& \dfrac{y_{i}(t)}{dt}= y_{i}(t) \left( x_{i+1,i+1}(t) - x_{i,i}(t) \right),\quad i=1,2,\ldots,m-1, \\
& x_{0,1}(t) \coloneqq 0, \quad x_{0,2}(t) \coloneqq 0, \quad \ldots, \quad x_{0,M-1}(t) \coloneqq 0, \\
& x_{m-M+2,m+1}(t) \coloneqq 0, \quad x_{m-M+3,m+1}(t) \coloneqq 0, \quad \ldots, \quad x_{m,m+1}(t) \coloneqq 0, \\
& y_{0}(t) \coloneqq 0, \quad y_{m}(t) \coloneqq 0.
\end{aligned}
\right.
\end{align}
Moreover, by introducing new variables given using the extended Toda variables $x_{i,i}(t)$ and $y_{i}(t)$ as:
\begin{align}
\label{eqn:KostantToda_extended_qdToda}
\left\{ 
\begin{aligned}
& \check{x}_{i}(t) = x_{i,i}(t),\quad i=1,2,\ldots,m, \\ 
& \check{y}_{i,j}(t) = y_{i}(t) y_{i+1}(t) \cdots y_{j-1}(t) x_{i,j}(t),\quad i=1,2,\ldots,m,\quad j=i+1,i+2,\ldots,\ell_{i}, \\ 
& \check{y}_{i,\ell_{i+1}}(t) = y_{i}(t) y_{i+1}(t) \cdots y_{\ell_i}(t),\quad i=1,2,\ldots,m, \\
& \check{y}_{0,1}(t) \coloneqq 0, \quad \check{y}_{0,2}(t) \coloneqq 0, \quad \ldots, \quad \check{y}_{0,M}(t) \coloneqq 0, \\
& \check{y}_{m-M+1,m+1}(t) \coloneqq 0, \quad \check{y}_{m-M+2,m+1}(t) \coloneqq 0, \quad \ldots, \quad \check{y}_{m,m+1}(t) \coloneqq 0, \\
\end{aligned}
\right.
\end{align}
we obtain the following theorem for a continuous analogue of the extended $q$-discrete Toda equation \eqref{eqn:extended_qdToda}. 
\begin{theorem}
\label{thm:KostantToda}
The continuous dynamics with respect to $\check{x}_i(t)$ and $\check{y}_{i,j}(t)$ satisfy the matrix representation:
\begin{align}
\label{eqn:Lax_KToda}
&\dfrac{d\check{A}(t)}{dt}=\check{A}(t)\check{A}_{\leq}(t)-\check{A}_{\leq}(t)\check{A}(t), 
\\
&\check{A}(t) \coloneqq
\left( \begin{array}{cccccc}
\check{x}_{1}(t) & \cdots & \check{y}_{1,\ell_{1}}(t) &\check{y}_{1,\ell_{2}}(t) & & \\
1 & \check{x}_{2}(t) & \cdots & \check{y}_{2,\ell_{2}}(t) & \ddots & \\
 & 1 & \check{x}_{3}(t) & \ddots & \ddots & \check{y}_{m-M,m}(t)  \\
 & & \ddots & \ddots & \ddots & \check{y}_{m-M+1,m}(t) \\
 & & & \ddots & \ddots & \vdots \\
 & & & & 1 & \check{x}_{m}(t) \\
\end{array} \right),
\nonumber
\\
&\check{A}_{\le}(t) \coloneqq
\left( \begin{array}{cccc}
\check{x}_{1}(t) & & \\
1 & \check{x}_{2}(t) &  \\
 & \ddots & \ddots  \\
 & & 1 & \check{x}_{m}(t) \\
\end{array} \right).
\nonumber
\end{align}
\end{theorem}
\begin{proof}
Taking the limit as $q\to1$ in the matrix representation \eqref{eqn:similarity_transformation}, we derive:
\begin{align}
\label{eqn:Lax_continuous_extended_qdToda}
&\dfrac{dA(t)}{dt}=A(t)B(t)-B(t)A(t), 
\\
&A(t) \coloneqq
\left( \begin{array}{cccccc}
x_{1,1}(t) & \cdots & x_{1,\ell_{1}}(t) & 1 & & \\
y_{1}(t) & x_{2,2}(t) & \cdots & x_{2,\ell_{2}}(t) & \ddots & \\
 & y_{2}(t) & x_{3,3}(t) & \ddots & \ddots & 1 \\
 & & \ddots & \ddots & \ddots & x_{m-M+1,m}(t) \\
 & & & \ddots & \ddots & \vdots \\
 & & & & y_{m-1}(t) & x_{m,m}(t) \\
\end{array} \right), 
\nonumber
\\
&B(t) \coloneqq
\left( \begin{array}{cccc}
0 & & & \\
y_{1}(t) & 0 & & \\
 & \ddots & \ddots & \\
 & & y_{m-1}(t) & 0
\end{array} \right). 
\nonumber
\end{align}
We can easily check that \eqref{eqn:Lax_continuous_extended_qdToda} is just the matrix representation of \eqref{eqn:continuous_extended_qdToda}. 
Preparing $D(t)\coloneqq{\rm diag}(1,y_1(t),y_1(t)y_2(t),\ldots,$ $\prod_{i=1}^{m-1}y_{i}(t))$ and considering \eqref{eqn:KostantToda_extended_qdToda}, we see that $\check{A}(t)=D(t)^{-1}A(t)D(t)$. 
Considering the derivative with respective to $t$ of $\check{A}(t)$ and using \eqref{eqn:Lax_continuous_extended_qdToda}, we obtain:
\[
\dfrac{d\check{A}(t)}{dt}=\dfrac{dD(t)^{-1}}{dt}A(t)D(t)+D(t)^{-1}\left(A(t)B(t)-B(t)A(t)\right)D(t)+D(t)^{-1}A(t)\dfrac{dD(t)}{dt}.
\]
Noting that $dD(t)^{-1}/dt=-D(t)^{-1}(dD(t)/dt)D(t)^{-1}$, we can rewrite this as:
\begin{align}
\label{eqn:Lax_continuous_KToda}
\begin{aligned}
\dfrac{d\check{A}(t)}{dt}
&=D(t)^{-1}A(t)D(t) \left( D(t)^{-1}B(t)D(t) + D(t)^{-1}\dfrac{dD(t)}{dt} \right)\\
&\quad -\left( D(t)^{-1}B(t)D(t) + D(t)^{-1}\dfrac{dD(t)}{dt} \right) D(t)^{-1}A(t)D(t).
\end{aligned}
\end{align}
Obviously, 
\begin{align*}
& (D(t))^{-1}B(t)D(t)=
\left( \begin{array}{cccc}
0 & & & \\
1 & 0 & & \\
 & \ddots & \ddots & \\
 & & 1 & 0
\end{array} \right),
\\
& (D(t))^{-1}\dfrac{dD(t)}{dt}=
{\rm diag}\left(0,\dfrac{1}{y_1(t)}\dfrac{dy_1(t)}{dt},\dots,
\dfrac{1}{y_1(t)y_2(t)\cdots  y_{m-1}(t)} \dfrac{d(y_1(t)y_2(t)\cdots  y_{m-1}(t))}{dt} 
\right).
\end{align*}
Here, by using the third equation of \eqref{eqn:continuous_extended_qdToda} and the first equation of \eqref{eqn:KostantToda_extended_qdToda}, we obtain: 
\begin{align*}
\dfrac{1}{y_{1}(t)\cdots y_{i-1}(t) } \dfrac{d( y_{1}(t)\cdots y_{i-1}(t) )}{dt} 
&= \sum_{j=1}^{i-1} \left( x_{j+1,j+1}(t)-x_{j,j}(t) \right),\\
&= \check{x}_{i}(t)-\check{x}_{1}(t).
\end{align*}
Thus, it follows that:
\begin{equation}
\label{eqn:KToda_proof_eq}
(D(t))^{-1}B(t)D(t) + D(t)^{-1}\dfrac{dD(t)}{dt}
=\check{A}_{\le}(t)-\check{x}_1(t)I.
\end{equation}
Considering $(D(t))^{-1}A(t)D(t)=\check{A}(t)$ and \eqref{eqn:KToda_proof_eq} in \eqref{eqn:Lax_continuous_KToda}, we therefore have \eqref{eqn:Lax_KToda}.
\end{proof}
\noindent
Theorem \ref{thm:KostantToda} implies that a continuous analogue of the extended $q$-discrete Toda equation \eqref{eqn:extended_qdToda} is the Kostant Toda equation. 
%
% ---------------------------------------------------------------------
\section{Interpretation of Hessenberg $\mbox{\boldmath $LR$}$ transformations}
\label{sec:3}
% ---------------------------------------------------------------------
%
In this section, we relate the extended $q$-discrete Toda equation \eqref{eqn:extended_qdToda} to shifted $LR$ transformations of Hessenberg matrices. 
% ---------------------------------------------------------------------
\par
We first show that a similarity transformation on a Hessenberg matrix by a lower bidiagonal matrix $L$ is uniquely determined by the first subdiagonal entry of $L$. This is an analogue of the famous implicit $Q$ theorem \cite{Golub_1996} on the uniqueness of a $QR$ transformation that preserves the Hessenberg structure, and we call it the implicit $L$ theorem.
\begin{theorem}[Implicit $L$ theorem]
\label{thm:implicit_L_theorem}
Let us assume that $A\coloneqq(a_{i,j})$ and $A^\ast\coloneqq(a_{i,j}^\ast)$ are $m$-by-$m$ Hessenberg matrices with $a_{2,1}\ne0,a_{3,2}\ne0,\ldots,a_{m-1,m-2}\ne0$ and $a_{m,m-1}\ne0$. 
If $L\coloneqq(\ell_{i,j})$ is a lower bidiagonal matrix with $1$ on all diagonal entries that satisfies:
\begin{equation}
\label{eqn:implicit_L}
LA^\ast=AL,
\end{equation}
then $L$ and $A^\ast$ are uniquely determined from $A$ and $\ell_{2,1}$.
Moreover, it holds that $a_{2,1}^\ast\neq0,a_{3,2}^\ast\neq0,\ldots,a_{m-1,m-2}^\ast\neq0$.
\end{theorem}
\begin{proof}
The matrix equation \eqref{eqn:implicit_L} immediately leads to:
\begin{align}
\label{eqn:implicit_L_lower}
\left\{
\begin{aligned}
&\ell_{i+1,i}a_{i,i}^\ast+a_{i+1,i}^\ast=a_{i+1,i}+a_{i+1,i+1}\ell_{i+1,i},\quad i=1,2,\ldots,m-1,\\
&\ell_{i+2,i+1}a_{i+1,i}^\ast=a_{i+2,i+1}\ell_{i+1,i},\quad i=1,2,\ldots,m-2.
\end{aligned}
\right.
\end{align}
\par
We first consider the case where $\ell_{2,1}=0$. If $\ell_{i+1,i}=0$ for some $i$, then it follows from the first equation of \eqref{eqn:implicit_L_lower} that $a^\ast_{i+1,i}=a_{i+1,i}\ne0$. Combining this with the second equation of \eqref{eqn:implicit_L_lower}, we obtain $\ell_{i+2,i+1}=0$. Thus, for $i=1,2,\dots,m-1$, we recursively have $\ell_{3,2}=0,\ell_{4,3}=0,\ldots,\ell_{m,m-1}=0$ and $a^\ast_{2,1}=a_{2,1}\ne0,a^\ast_{3,2}=a_{3,2}\ne0,\dots,a^\ast_{m,m-1}=a_{m,m-1}\ne0$. 
\par
Next, we consider the case where $\ell_{2,1}\ne0$. If $\ell_{i+1,i}\ne 0$ for some $i$, then we have $\ell_{i+2,i+1}\neq0$ and $a_{i+1,i}^\ast\ne0$ from the second equation of \eqref{eqn:implicit_L_lower}. 
Thus, for $i=1,2,\dots,m-1$, we recursively have $\ell_{3,2}\ne0,\ell_{4,3}\ne0,\ldots,\ell_{m,m-1}\ne0$, and $a^\ast_{2,1}\ne0,a^\ast_{3,2}\ne0,\dots,a^\ast_{m,m-1}\ne0$. 
\par
Let $\bm{\ell}_i$ and $\bm{a}_i^{\ast}$ denote the $i$th columns of $L$ and $A^{\ast}$, respectively.
Then, we can rewrite the matrix equation \eqref{eqn:implicit_L} as:
\begin{equation}
\label{eqn:implicit_L_vector}
L\bm{a}_i^\ast=A\bm{\ell}_i,\quad i=1,2,\ldots,m.
\end{equation}
Furthermore, let $(\cdot)_{1:i}^{1:i}$ be the $i$-by-$i$ principal submatrix of a matrix and let $(\cdot)_{1:i}$ be the vector consisting of the $1$st, $2$nd, $\ldots,$ $i$th entries of a vector. Observing the first $i+1$ entries of \eqref{eqn:implicit_L_vector}, we obtain:
\begin{equation}
\label{eqn:implicit_L_vector_i_principal}
(L)_{1:i+1}^{1:i+1} (\bm{a}_i^\ast)_{1:i+1} = (A\bm{\ell}_i)_{1:i+1},\quad i=1,2,\dots,m-1.
\end{equation}
Since $\det(L)_{1:i+1}^{1:i+1}=1$, the inverse matrix $((L)_{1:i+1}^{1:i+1})^{-1}$ exists. 
Thus, it follows from \eqref{eqn:implicit_L_vector_i_principal} that:
\[
(\bm{a}_i^\ast)_{1:i+1} = ((L)_{1:i+1}^{1:i+1})^{-1} (A\bm{\ell}_i)_{1:i+1},\quad i=1,2,\dots,m-1.
\]
This shows that the nonzero part of the $i$th column vector of $A^\ast$, denoted by $(\bm{a}_i^\ast)_{1:i+1}$, is uniquely determined if $\bm{\ell}_1,\bm{\ell}_2,\dots,\bm{\ell}_i$, and $A$ are given.
Moreover, by noting that $\ell_{i+2,i+1}$ is uniquely determined from $a_{i+2,i+1}$, $\ell_{i+1,i}$, and $a^{\ast}_{i+1,i}$ as:
\[
\ell_{i+2,i+1}=\dfrac{a_{i+2,i+1}\ell_{i+1,i}}{a_{i+1,i}^\ast},
\]
which is easily derived from the second equation of \eqref{eqn:implicit_L_lower}, we see that $\bm{a}_i^{\ast}$ and $\bm{\ell}_{i+1}$ are uniquely determined if $\bm{\ell}_1,\bm{\ell}_2,\dots,\bm{\ell}_i$ and $A$ are given. 
Thus, by induction for $i=1,2,\dots,m-1$, we can conclude that $L$ and $A^{\ast}$ are uniquely determined for given $\bm{\ell}_1$, to be precise $\ell_{2,1}$, and $A$. 
\end{proof}
\par
Now we apply this theorem to reinterpret our extended $q$-Toda equation \eqref{eqn:extended_qdToda} as a shifted $LR$ transformation. To this end, we consider the $LR$ factorization of a Hessenberg matrix $A^{(n)}+(1/\mu^{(n)})I$ as:
\begin{equation}
\label{eqn:implicit_LR}
A^{(n)}+\dfrac{1}{\mu^{(n)}}I = L^{(n)}R^{(n)},
\end{equation}
where $L^{(n)}$ is a lower bidiagonal matrix with $1$ on the diagonal entries and $R^{(n)}$ is an upper triangular matrix with $1$ on the $(M+1)$th upper diagonal entries. It is easy to see from \eqref{eq:An} that the first subdiagonal entry of $L^{(n)}$ is:
\begin{equation}
\ell_{2,1}^{(n)} = \frac{\mu^{(n)}y_1^{(n)}}{1+\mu^{(n)}x_{1,1}^{(n)}}.
\end{equation}
We also define a new $m$-by-$m$ matrix as:
\begin{equation}
\label{eqn:implicit_RL}
\bar{A}^{(n+1)} = R^{(n)}L^{(n)}-\dfrac{1}{\mu^{(n)}}I.
\end{equation}
Obviously, $\bar{A}^{(n+1)}$ is a Hessenberg matrix with the same form as $A^{(n)}$. The transformation from $A^{(n)}$ to $\bar{A}^{(n+1)}$ is called a shifted $LR$ transformation with shift $-1/\mu^{(n)}$. It follows from \eqref{eqn:implicit_LR} and \eqref{eqn:implicit_RL} that $L^{(n)}\bar{A}^{(n+1)}=A^{(n)}L^{(n)}$. 
We here recall that $I+\mu^{(n)}G^{(n)}$ is also a bidiagonal matrix with $1$ on the diagonal entries satisfying $(I+\mu^{(n)}G^{(n)})A^{(n+1)}=A^{(n)}(I+\mu^{(n)}G^{(n)})$ and $(I+\mu^{(n)}G^{(n)})_{2,1}=\mu^{(n)}y_1^{(n)}/(1+\mu^{(n)}x_{1,1}^{(n)})$. 
Thus, if the lower subdiagonal entries of $A^{(n)}$ are nonzero, then, from Theorem \ref{thm:implicit_L_theorem}, we have:
\begin{align*}
\left\{
\begin{aligned}
&\bar{A}^{(n+1)}=A^{(n+1)},\\
&L^{(n)}=I+\mu^{(n)}G^{(n)}.
\end{aligned}
\right.
\end{align*}
We therefore see that the extended $q$-discrete Toda equation \eqref{eqn:extended_qdToda} generates a shifted $LR$ transformation from $A^{(n)}$ to $A^{(n+1)}$. 
\begin{theorem}
\label{thm:LR_transformation}
Let us assume that the lower subdiagonal entries $y_1^{(n)}, \ldots, y_{m-1}^{(n)}$ of $A^{(n)}$ are nonzero and the $LR$ factorization of $A^{(n)}+(1/\mu^{(n)})I$ is given as $A^{(n)}+(1/\mu^{(n)})I=L^{(n)}R^{(n)}$. Then, the extended $q$-discrete Toda equation \eqref{eqn:extended_qdToda} can be expressed in matrix form as:
\begin{equation}
\label{eqn:LR_transformation}
L^{(n+1)} R^{(n+1)} - \frac{1}{\mu^{(n+1)}} I = R^{(n)} L^{(n)} - \frac{1}{\mu^{(n)}} I.
\end{equation}
\end{theorem}
%
%
% ---------------------------------------------------------------------
\section{Asymptotic convergence to matrix eigenvalues}
\label{sec:4}
% ---------------------------------------------------------------------
%
In this section, we consider the case where the initial matrix $A^{(0)}$ is a nonsingular TN matrix and show asymptotic convergence as $n\to\infty$ to matrix eigenvalues in the extended $q$-discrete Toda equation \eqref{eqn:extended_qdToda}. We first prove that time evolution of the extended $q$-discrete Toda equation never breaks down if the initial matrix is nonsingular TN and the TN property is retained throughout the evolution. We then use the convergence theorem of $LR$ transformations of a TN Hessenberg matrix. Finally, we present two examples to verify the convergence numerically. 
\par
Let us begin by preparing a lemma showing the positivity of entries of bidiagonal matrices in the tridiagonal $LU$ decomposition.
\begin{lemma}
\label{lem:TN_decomposition}
Let $L$ and $U$ be $m$-by-$m$ lower and upper bidiagonal matrices given by:
\[
L\coloneqq
\left( \begin{array}{cccc}
1 & & & \\
e_{1} & 1 & & \\
 & \ddots & \ddots & \\
 & & e_{m-1} & 1
\end{array} \right),\quad
U\coloneqq
\left( \begin{array}{cccc}
q_1 & f_{1} & & \\
 & q_2 & \ddots & \\
 & & \ddots & f_{m-1} \\
 & & &q_{m} 
\end{array} \right). 
\]
If $e_1>0,e_2>0,\ldots,e_{m-1}>0$, $q_1>0,q_2>0,\ldots,q_m>0$, and $f_1\geq0,f_2\geq0,\ldots,f_{m-1}\geq0$, then there exists an $LU$ decomposition of $UL$ such that
\[
\bar{L} \bar{U}=UL,
\]
where $\bar{L}$ and $\bar{U}$ have the same forms as $L$ and $U$, respectively. 
Moreover, it holds that $\bar{e}_1>0,\bar{e}_2>0,\ldots,\bar{e}_{m-1}>0$, $\bar{q}_1>0,\bar{q}_2>0,\ldots,\bar{q}_m>0$, and $\bar{f}_1=f_1,\bar{f}_2=f_2,\ldots,\bar{f}_{m-1}=f_{m-1}$.
\end{lemma}
\begin{proof}
Observing the $(i,i+1),(i,i)$ and $(i-1,i)$ entries in the matrix equality $\bar{L}\bar{U}=UL$, we derive:
\begin{align}
\label{eqn:TN_decomposition_proof_1}
\left\{
\begin{aligned}
& \bar{f}_i = f_i,\quad i=1,2,\ldots,m-1,\\
& \bar{q}_i+\bar{e}_{i-1} \bar{f}_{i-1}=q_i+e_i f_i,\quad i=1,2,\ldots,m, \\
& \bar{q}_i \bar{e}_i =q_{i+1}e_i,\quad i=1,2,\ldots,m-1,
\end{aligned}
\right.
\end{align}
where $\bar{e}_0\coloneqq0,\bar{f}_0\coloneqq0$ and $\bar{e}_m\coloneqq0,\bar{f}_m\coloneqq0$. 
Thus, it holds that $\bar{f}_1=f_1,\bar{f}_2=f_2,\ldots,\bar{f}_{m-1}=f_{m-1}$. 
By introducing new variables:
\begin{equation}
\label{eqn:d_positivity}
d_i=q_i-\bar{e}_{i-1} \bar{f}_{i-1},\quad i=1,2,\ldots,m,
\end{equation}
and using the first and third equations of \eqref{eqn:TN_decomposition_proof_1}, we obtain:
\[
d_i=\dfrac{q_{i}}{\bar{q}_{i-1}}( \bar{q}_{i-1}-e_{i-1} f_{i-1} ),\quad i=1,2,\ldots,m,
\]
As $\bar{q}_{i-1}-e_{i-1}f_{i-1}=d_{i-1}$, we can rewrite \eqref{eqn:TN_decomposition_proof_1} as:
\begin{align*}
\left\{
\begin{aligned}
& \bar{f}_i = f_i,\quad i=1,2,\ldots,m-1,\\
& \bar{q}_i=d_i+e_i f_i,\quad i=1,2,\ldots,m, \\
& \bar{e}_i=\dfrac{q_{i+1}}{\bar{q}_i}e_i,\quad i=1,2,\ldots,m-1,\\
& d_1\coloneqq q_1,\quad d_{i}=\dfrac{q_{i}}{\bar{q}_{i-1}}d_{i-1},\quad i=2,3,\ldots,m.
\end{aligned}
\right.
\end{align*}
This suggests that the $LU$ decomposition of $UL$ is uniquely given, and the positivities of $q_i$ and $e_i$ dominate those of $\bar{q}_i$ and $\bar{e}_i$.
\end{proof}
%
% ---------------------------------------------------------------------
\par
Using Lemma \ref{lem:TN_decomposition}, we show that the shifted $LR$ transformation from $A^{(n)}$ to $A^{(n+1)}$ does not break down if $A^{(n)}$ is nonsingular TN and that the TN property is inherited by $A^{(n+1)}$. 
\begin{proposition}
\label{prop:TN_properties}
Let us assume that $A^{(n)}$ is a nonsingular TN Hessenberg matrix with positive subdiagonals, namely, $y_{1}^{(n)}>0,y_{2}^{(n)}>0,\ldots,y_{m-1}^{(n)}>0$. If $1/\mu^{(n)}>-\lambda_{\min}(A^{(n)})$, then $A^{(n)}+(1/\mu^{(n)})I$ admits the $LU$ decomposition. The Hessenberg matrix $A^{(n+1)}$ obtained by the extended $q$-discrete Toda equation \eqref{eqn:extended_qdToda} is also a nonsingular TN matrix with positive subdiagonals. 
\end{proposition}
\begin{proof}
If the leading principal minors of $A^{(n)}+(1/\mu^{(n)})I$ are all positive, then its $LU$ factorization is uniquely determined. 
Since $A^{(n)}$ is nonsingular TN, all eigenvalues of $A^{(n)}$ are real and positive. Combining this fact with the interlacing theorem \cite{Li_2002}, we obtain:
\[
\lambda_{\min}( (A^{(n)})_{1:i}^{1:i} ) \geq \lambda_{\min}( A^{(n)} )>0.
\]
If $1/\mu^{(n)}>-\lambda_{\min}( A^{(n)} )$, then $\lambda_k((A^{(n)})_{1:i}^{1:i})+1/\mu^{(n)}>0$ for $k=1,2,\ldots,i$. 
Thus, it follows that:
\begin{align}
\label{eqn:det_positivity}
\begin{aligned}
\det\left( \left( A^{(n)}+\dfrac{1}{\mu^{(n)}} I \right)_{1:i}^{1:i} \right)
&=\prod_{k=1}^i \lambda_k\left(  \left( A^{(n)}+\dfrac{1}{\mu^{(n)}}I \right)_{1:i}^{1:i} \right) \\
&=\prod_{k=1}^i \left( \lambda_k((A^{(n)})_{1:i}^{1:i})+\dfrac{1}{\mu^{(n)}} \right) \\
&> 0.
\end{aligned}
\end{align}
We therefore see that $A^{(n)}+(1/\mu^{(n)})I$ admits the $LU$ decomposition. 
According to \cite{Pinkus_2009}, the $(i,i)$ entry of $R^{(n)}$ is given as:
\[
r_{i,i}^{(n)} = \dfrac{\det\left( \left( A^{(n)}+(1/\mu^{(n)})I \right)_{1:i}^{1:i} \right) }{\det\left( \left( A^{(n)}+(1/\mu^{(n)})I \right)_{1:i-1}^{1:i-1} \right) },\quad i=1,2,\ldots,m,
\]
where $\det\left( \left( A^{(n)}+(1/\mu^{(n)})I \right)_{1:0}^{1:0} \right)\coloneqq1$. Combining this with \eqref{eqn:det_positivity}, we derive $r_{i,i}^{(n)}>0$. Observing the $(i+1,i)$ entries on the $LU$ factorization $A^{(n)}+(1/\mu^{(n)})I=L^{(n)}R^{(n)}$, we obtain:
\[
y_{i}^{(n)}=\ell_{i+1,i}^{(n)}r_{i,i}^{(n)},\quad i=1,2,\ldots,m-1.
\]
This result combined with $r_{i,i}^{(n)}>0$ and $y_{i}^{(n)}>0$ leads to $\ell_{i,i+1}^{(n)}>0$. Thus, $L^{(n)}$ is a positive lower bidiagonal matrix.
Focusing on the equalities of the $(i+1,i)$ entries of $A^{(n+1)}=R^{(n)}L^{(n)}-(1/\mu^{(n)})I$, namely,
\[
y_{i}^{(n+1)}=r_{i+1,i+1}^{(n)}\ell_{i+1,i}^{(n)},\quad i=1,2,\ldots,m-1,
\]
we obtain $y_{i}^{(n+1)}>0$. 
\par
It remains to be shown that $A^{(n+1)}$ is TN. Since $A^{(n)}$ is nonsingular TN, it admits the $LDU$ decomposition \cite[Theorem 4.1]{Gaska_2010}:
\begin{equation}
A^{(n)}={\cal L} {\cal D} {\cal U}, \label{eq:LDU}
\end{equation}
where ${\cal L}$ is a product of unit lower bidiagonal matrices of the form $I+c{\bm e}_{i+1}{\bm e}_i^{\top}$ ($c>0$), where ${\bm e}_i$ is the $i$th column of the identity matrix of order $m$, ${\cal D}$ is a diagonal matrix with positive diagonals, and ${\cal U}$ is a product of unit upper bidiagonal matrices of the form $I+c{\bm e}_i{\bm e}_{i+1}^{\top}$ ($c>0$). Note that ${\cal L}$ is actually a unit lower bidiagonal matrix with positive subdiagonal entries because it is also characterized as the lower triangular factor of the $LU$ decomposition of $A^{(n)}$. Furthermore, we can write ${\cal U}$ as:
\[
{\cal U}={\cal U}_1 {\cal U}_2 \cdots {\cal U}_N,
\]
where $N$ is the number of factors of ${\cal U}$ and each of ${\cal U}_1,{\cal U}_2,\dots,{\cal U}_N$ is a unit upper bidiagonal matrix with only one positive subdiagonal entries.
Thus, by repeatedly applying Lemma \ref{lem:TN_decomposition}, we derive:
\begin{align}
{\cal U}_1 \cdots {\cal U}_{N-1} {\cal U}_N L^{(n)}
&= {\cal U}_1 \cdots {\cal U}_{N-1} L^{(n,1)} \bar{{\cal U}}_N \nonumber \\
&= {\cal U}_1 \cdots L^{(n,2)} \bar{{\cal U}}_{N-1} \bar{{\cal U}}_N \nonumber \\
&\,\,\,\vdots \nonumber \\
&=  L^{(n,N)} \bar{{\cal U}}_1 \bar{{\cal U}}_2 \cdots \bar{{\cal U}}_N,
\label{eqn:unit_decomposition}
\end{align}
where $L^{(n,1)},L^{(n,2)},\ldots,L^{(n,N)}$ are unit bidiagonal matrices with positive subdiagonals, and $\bar{\cal U}_i$ has the same structure as ${\cal U}_i$. Then, by introducing auxiliary matrices:
\begin{align*}
\left\{
\begin{aligned}
& D_{N+1} \coloneqq I,\quad D_i\coloneqq (\bar{{\cal U}}_iD_{i+1})_{\rm diag}, \quad i=N,N-1,\ldots,1,\\
& {\cal U}_i^\ast \coloneqq D_i^{-1} \bar{{\cal U}}_i D_{i+1},\quad i=N,N-1,\ldots,1,
\end{aligned}
\right.
\end{align*}
where $(\cdot)_{\rm diag}$ denotes the diagonal matrix consisting of the diagonal entries of a matrix, and noting that $L^{(n,N+1)}\coloneqq{\cal D} L^{(n,N)} {\cal D}^{-1}$, we can rewrite \eqref{eqn:unit_decomposition} as:
\[
{\cal U}_1 {\cal U}_2 \cdots {\cal U}_N L^{(n)}
= {\cal D}^{-1} L^{(n,N+1)} {\cal D} D_1 {\cal U}_1^\ast {\cal U}_2^\ast \cdots {\cal U}_N^\ast.
\]
Since $A^{(n+1)}=(L^{(n)})^{-1}A^{(n)}L^{(n)}$ with $A^{(n)}={\cal L} {\cal D} {\cal U}_1 {\cal U}_2 \cdots {\cal U}_N$, it follows that:
\[
A^{(n+1)} = (L^{(n)})^{-1} {\cal L} L^{(n,N+1)} {\cal D} D_1 {\cal U}_1^\ast {\cal U}_2^\ast \cdots {\cal U}_N^\ast.
\]
Putting ${\cal L}^\ast\coloneqq(L^{(n)})^{-1} {\cal L} L^{(n,N+1)}$, ${\cal D}^\ast\coloneqq{\cal D} D_1$, and ${\cal U}^{\ast}\coloneqq{\cal U}_1^{\ast}{\cal U}_2^{\ast}\cdots {\cal U}_N^{\ast}$, we obtain the $LDU$ decomposition of $A^{(n+1)}$ as $A^{(n+1)}={\cal L}^{\ast}{\cal D}^{\ast}{\cal U}^{\ast}$. 
Since $A^{(n+1)}$ is a Hessenberg matrix, ${\cal L}^\ast$ is unit lower bidiagonal. By focusing on the $(i+1,i)$ entry of $A^{(n+1)}={\cal L}^\ast {\cal D}^\ast {\cal U}^{\ast}$, we have:
\begin{equation}
\label{eqn:positivity_cal_L}
y_{i}^{(n+1)}=\ell_{i+1,i}^\ast d_{i,i}^\ast,
\end{equation}
where $\ell^\ast_{i+1,i}$ denotes the $(i+1,i)$ entry of ${\cal L}^\ast$, and $d_i^\ast$ denotes the $i$th diagonal entry of ${\cal D}^\ast$. Equation \eqref{eqn:positivity_cal_L} with $y_{i}^{(n+1)}>0$ and $d_i^\ast>0$ leads to the positivity $\ell_{i+1,i}^{\ast}>0$, which implies that ${\cal L}^\ast$ is TN. Since ${\cal D}^\ast$ and ${\cal U}^\ast$ are products of nonnegative (bi)diagonal matrices and are therefore TN, we conclude that $A^{(n+1)}={\cal L}^{\ast}{\cal D}^{\ast}{\cal U}^{\ast}$ is also TN.
\end{proof}
\noindent
Using Proposition \ref{prop:TN_properties} repeatedly, we obtain a theorem concerning the TN property of the matrix sequence $\{A^{(n)}\}_{n=0,1,\dots}$.
\begin{theorem}
\label{thm:TN_properties}
Let us assume that $A^{(0)}$ is nonsingular TN matrix with positive subdiagonals. If $1/\mu^{(0)}>-\lambda_{\min}( A^{(0)} ),1/\mu^{(1)}>-\lambda_{\min}( A^{(0)} ),\ldots,$ then $A^{(n)}+(1/\mu^{(n)})I$ allows the $LR$ decomposition for any $n\ge 0$ and time evolution of the extended $q$-Toda equation \eqref{eqn:extended_qdToda} can be carried out without breakdown. Moreover, the Hessenberg matrix $A^{(n)}$ is nonsingular TN with positive subdiagonals.
\end{theorem}
Note that while we used the bidiagonal decomposition \eqref{eq:LDU} of $A^{(n)}$ in the proof of Proposition \ref{prop:TN_properties}, it is only for theoretical purposes and not computed in practice. In the actual time evolution, the recursion formula \eqref{eqn:extended_qdToda} is used.
%
% ---------------------------------------------------------------------
\par
According to Fallat et al.~\cite{Fallat_2005}, a nonsingular and irreducible TN matrix has distinct positive eigenvalues. Our $A^{(0)}$ falls in this category because a nonsingular TN Hessenberg matrix with positive subdiagonals is irreducible. Hereinafter, let $\lambda_1,\lambda_2,\dots,\lambda_m$ denote eigenvalues of $A^{(0)}$ such that $\lambda_1>\lambda_2>\cdots >\lambda_m>0$. Furthermore, $A^{(0)}$ is diagonalizable since it has $m$ distinct eigenvalues. Fukuda et al.~\cite{Fukuda_2013_2} presents a convergence theorem for the shifted $LR$ transformations of TN Hessenberg matrices. While they dealt with a TN Hessenberg matrix expressed as a product of positive bidiagonal matrices, a closer examination reveals that the convergence theorem given by them is applicable to a general nonsingular TN Hessenberg matrix with positive subdiagonal entries (and therefore with distinct positive eigenvalues). In the notation of our paper, the following lemma and theorem are proved in \cite{Fukuda_2013_2}.
\begin{lemma}(cf.~\cite[Lemma 1]{Fukuda_2013_2})
\label{lemma:Fukuda_lemma1}
Let $A^{(0)}$ be a nonsingular TN Hessenberg matrix with positive subdiagonal entries and denote its eigendecomposition by $A^{(0)}=PDP^{-1}$. Then, both $P$ and $P^{-1}$ admit the $LU$ decomposition.
\end{lemma}
\begin{theorem}(cf. \cite[Theorem 4]{Fukuda_2013_2})
\label{thm:convergence_LR_transformation}
Let $A^{(0)}$ be a matrix with distinct positive eigenvalues $\lambda_1>\lambda_2>\cdots >\lambda_m>0$ and assume that the matrices $P$ and $P^{-1}$ in its eigendecomposition $A^{(0)}=PDP^{-1}$ allow the $LU$ decomposition. Moreover, assume that the sequence of the shifted $LR$ transformations:
\begin{equation}
A^{(n)}+\frac{1}{\mu^{(n)}}I = L^{(n)}R^{(n)}, \quad A^{(n+1)}=R^{(n)}L^{(n)}-\frac{1}{\mu^{(n)}}I,
\end{equation}
where $\mu^{(n)}$ is chosen to satisfy $1/\mu^{(n)}>-\lambda_m$, does not break down. Then, the matrix sequence $\{A^{(n)}\}_{n=0,1,\dots}$ converges to an upper triangular matrix with $\lambda_1,\lambda_2,\ldots,\lambda_m$ on the diagonal as $n\to\infty$. 
\end{theorem}
\noindent
By combining Theorem \ref{thm:TN_properties}, Lemma \ref{lemma:Fukuda_lemma1}, and Theorem \ref{thm:convergence_LR_transformation}, and noting that $L^{(n)}=I+\mu^{(n)}G^{(n)}$ tends to the identity matrix as $n\to\infty$, we have the following convergence theorem for the extended $q$-discrete Toda equation \eqref{eqn:extended_qdToda}.
\begin{theorem}
\label{thm:convergence_extended_qToda}
Let us assume that the initial setting of the extended $q$-discrete Toda equation \eqref{eqn:extended_qdToda} is given from entries of the Hessenberg matrix $A^{(0)}$ with positive subdiagonals. Additionally, let $1/\mu^{(0)}>-\lambda_m,1/\mu^{(1)}>-\lambda_m,\ldots.$ Then, it holds that 
\begin{align*}
\begin{aligned}
& \lim_{n \to \infty} x_{i,i}^{(n)} = \lambda_{i}, \quad i=1,2,\ldots,m, \\
& \lim_{n \to \infty} x_{i,j}^{(n)} = c_{i,j}, \quad i=1,2,\ldots,m-1, \quad j=i+1,i+2,\ldots,\ell_{i}, \\
& \lim_{n \to \infty} y_{i}^{(n)} = 0, \quad i=1,2,\ldots,m-1, \\
& \lim_{n \to \infty} g_{i}^{(n)} = 0, \quad i=1,2,\ldots,m-1, 
\end{aligned}
\end{align*}
where $c_{i,j}$ are some constants. 
\end{theorem}
%
% ---------------------------------------------------------------------
\par
In the remainder of this section, we present two numerical examples to demonstrate the convergence to matrix eigenvalues in the extended $q$-discrete Toda equation \eqref{eqn:extended_qdToda}. 
We used floating-point arithmetic, the computer has a operating system Mac OS Monterey (ver. 12.5 (21G72)) and an Apple M1 CPU, and we employed the numerical computation software Maple 2022.1. 
% ---------------------------------------------------------------------
\begin{figure}[tb]
\begin{center}
\includegraphics[width=0.6\textwidth]{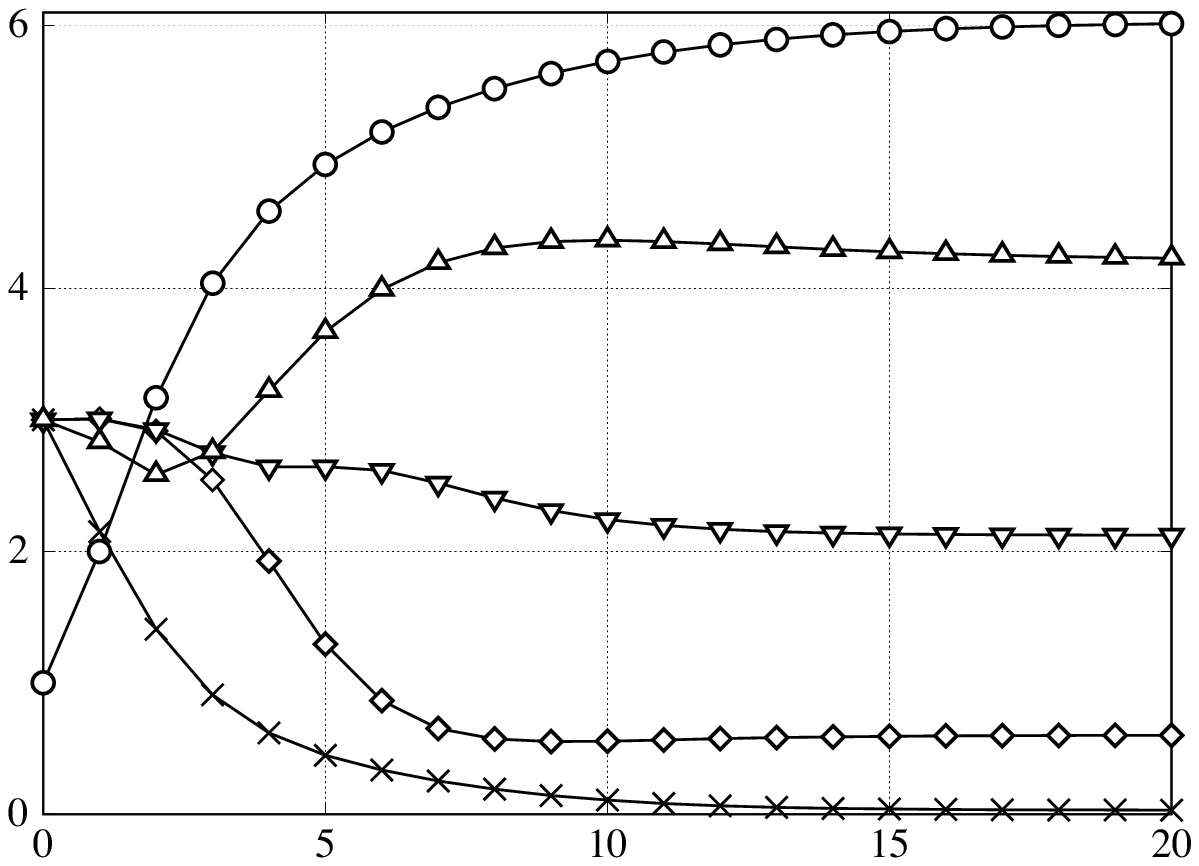}
\caption{
Discrete-time $n$ ($x$-axis) versus values of 
$x_{1,1}^{(n)}$, $x_{2,2}^{(n)}$, $x_{3,3}^{(n)}$ and $x_{4,4}^{(n)}$ ($y$-axis). 
Circles: $\{ x_{1,1}^{(n)} \}_{n=0,1,\dots,20}$; 
triangles: $\{ x_{2,2}^{(n)} \}_{n=0,1,\dots,20}$; 
downward~triangles: $\{ x_{3,3}^{(n)} \}_{n=0,1,\dots,20}$; 
diamonds: $\{ x_{4,4}^{(n)} \}_{n=0,1,\dots,20}$; and 
crosses: $\{ x_{5,5}^{(n)} \}_{n=0,1,\dots,20}$. 
\label{fig:fact1}
}
\end{center}
\end{figure}
% ---------------------------------------------------------------------
\par
The target matrix in the first example is  a $5$-by-$5$ Hessenberg matrix:
\[
A^{(0)} = 
\left( \begin{array}{ccccc}
1 & 2 & 1 & & \\
1 & 3 & 3 & 1 & \\
& 1 & 3 & 3 & 1 \\
& & 1 & 3 & 3 \\
& & & 1 & 3  
\end{array} \right), 
\]
which is given by products of bidiagonal matrices as:
\begin{equation}
\label{eqn:exam_factorization}
A^{(0)} =
\left( \begin{array}{ccccc}
1 & & & & \\
1 & 1 & & & \\
 & 1 & 1 & & \\
 & & 1 & 1 & \\
 & & & 1 & 1 
\end{array} \right)
\left( \begin{array}{ccccc}
1 & 1 & & & \\
 & 1 & 1 & & \\
 & & 1 & 1 & \\
 & & & 1 & 1 \\
 & & & & 1 
\end{array} \right)
\left( \begin{array}{ccccc}
1 & 1 & & & \\
 & 1 & 1 & & \\
 & & 1 & 1 & \\
 & & & 1 & 1 \\
 & & & & 1 
\end{array} \right). 
\end{equation}
Since the three bidiagonal matrices are TN, the Hessenberg matrix $A^{(0)}$ is also TN. 
The {\tt Eigenvalues} function in Maple returns 
$\lambda_{1}=6.03136292416233$, $\lambda_{2}=4.21379563011770$, $\lambda_{3}=2.12210018294618$, $\lambda_{4}=0.601938246298446$ and $\lambda_{5}=0.030803016475350$ as five eigenvalues of $A^{(0)}$. 
If the bidiagonal factorization \eqref{eqn:exam_factorization} is given, we can also compute the five eigenvalues using discrete-time evolutions in the dhToda equation \cite{Fukuda_2013_1}. 
From the matrix structure, we can set $m=5$ and $M=2$ in the extended $q$-discrete Toda equation \eqref{eqn:extended_qdToda}. Next, from the entries of $A^{(0)}$, we can directly determine the initial settings in the extended $q$-discrete Toda equation \eqref{eqn:extended_qdToda} as 
$x_{1,1}^{(0)}=1$, $x_{1,2}^{(0)}=2$, 
$x_{2,2}^{(0)}=x_{2,3}^{(0)}=x_{3,3}^{(0)}=x_{3,4}^{(0)}=x_{4,4}^{(0)}=x_{4,5}^{(0)}=x_{5,5}^{(0)}=3$, 
and $y_{1}^{(0)}=y_{2}^{(0)}=y_{3}^{(0)}=1$. 
We emphasize that the extended $q$-discrete Toda equation \eqref{eqn:extended_qdToda} does not require the bidiagonal factorization \eqref{eqn:exam_factorization} in the initial settings. 
We fix the shift parameters as $\mu^{(n)}=1$ for all $n$. 
Figure \ref{fig:fact1} plots the values of the extended $q$-discrete Toda variables $x_{1,1}^{(n)}$, $x_{2,2}^{(n)}$, $x_{3,3}^{(n)}$, $x_{4,4}^{(n)}$, and $x_{5,5}^{(n)}$ at $n=0,1,\dots,20$. 
The values of $\vert x_{1,1}^{(120)}-\lambda_{1}\vert/\lambda_{1}, \vert x_{2,2}^{(120)}-\lambda_{2}\vert/\lambda_{2}, \vert x_{3,3}^{(120)}-\lambda_{3}\vert/\lambda_{3}, \vert x_{4,4}^{(120)}-\lambda_{4}\vert/\lambda_{4}$, and $\vert x_{5,5}^{(120)}-\lambda_{5}\vert/\lambda_{5}$ are $1.4725998598790676\times10^{-16}, 6.323361389564953\times10^{-16}, 4.1853745965331373\times10^{-16}, 5.533240484978411\times10^{-16}$ and $1.9710836342205465\times10^{-14}$, respectively. 
Thus we observe that $x_{1,1}^{(n)}, x_{2,2}^{(n)}, x_{3,3}^{(n)}, x_{4,4}^{(n)}$, and $x_{5,5}^{(n)}$ respectively approach the eigenvalues $\lambda_1,\lambda_2,\lambda_3,\lambda_4$, and $\lambda_5$. 
Values of the extended $q$-discrete Toda variables $x_{1,2}^{(n)}, x_{2,3}^{(n)}, x_{3,4}^{(n)}$, and $x_{4,5}^{(n)}$ at $n=0,1,\dots,20$ are illustrated in Figure \ref{fig:fact2}. 
In the initial settings, we obviously find that values of the extended $q$-discrete Toda variables are sorted as $x_{1,2}^{(n)}>x_{2,3}^{(n)}>x_{3,4}^{(n)}>x_{4,5}^{(n)}$ as $n$ increases. 
Furthermore, since the extended $q$-discrete Toda variables $y_{1}^{(120)}, y_{2}^{(120)}, y_{3}^{(120)}$, and $y_{4}^{(120)}$ and $g_{1}^{(120)}, g_{2}^{(120)}, g_{3}^{(120)}$, and $g_{4}^{(120)}$ are, respectively, $7.630512286281922\times10^{-16}$, $2.8005008112011965\times10^{-26}$, $2.444684147920427\times10^{-34}$, and $5.982071888277229\times10^{-23}$ and $1.4635234726508965\times10^{-16}$, $8.969926162197963\times10^{-27}$, $1.5260788944700523\times10^{-34}$, and $5.803312361979572\times10^{-23}$, we also see that $y_{i}^{(n)}$ and $g_{i}^{(n)}$ converge to $0$. 
% ---------------------------------------------------------------------
\begin{figure}[tb]
\begin{center}
\includegraphics[width=0.6\textwidth]{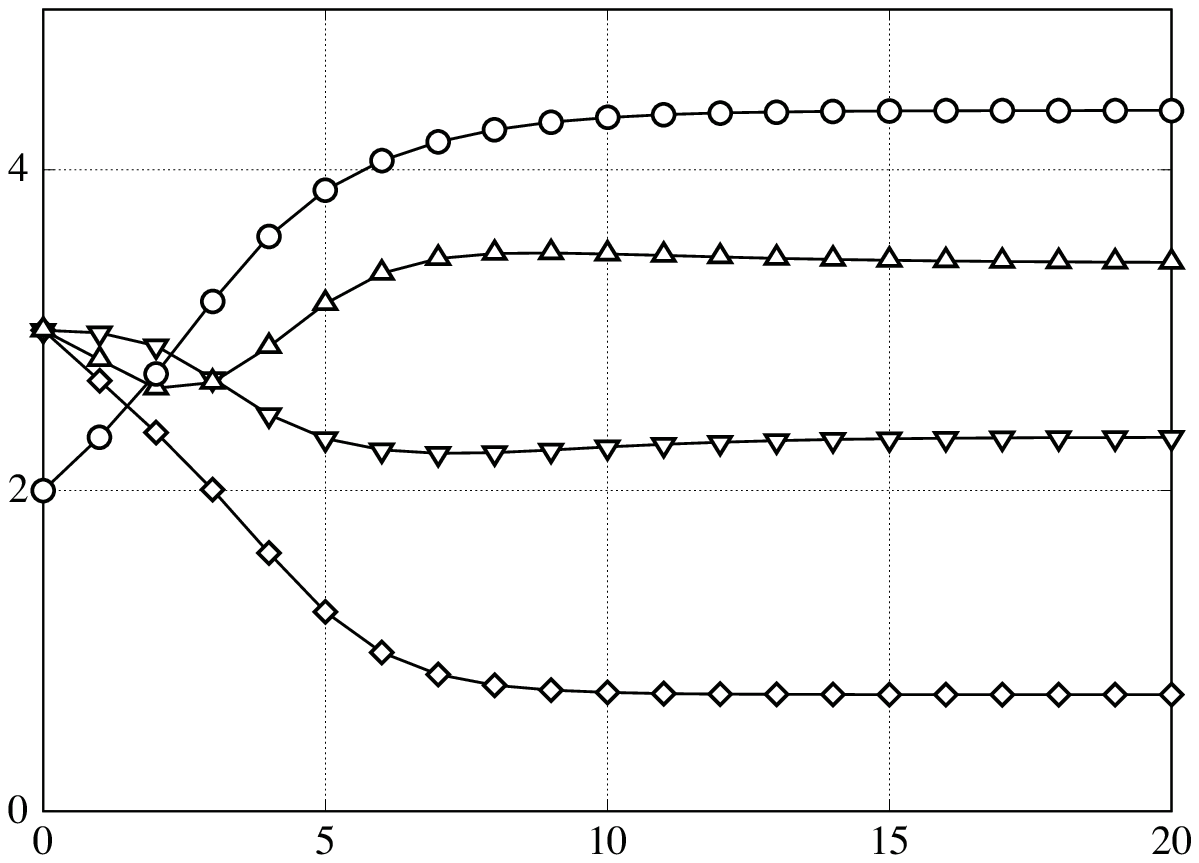}
\caption{
Discrete-time $n$ ($x$-axis) versus values of 
$x_{1,2}^{(n)}$, $x_{2,3}^{(n)}$, $x_{3,4}^{(n)}$ and $x_{4,5}^{(n)}$ ($y$-axis). 
Circles: $\{ x_{1,2}^{(n)} \}_{n=0,1,\dots,20}$; 
triangles: $\{ x_{2,3}^{(n)} \}_{n=0,1,\dots,20}$; 
downward~triangles: $\{ x_{3,4}^{(n)} \}_{n=0,1,\dots,20}$; and 
diamonds: $\{ x_{4,5}^{(n)} \}_{n=0,1,\dots,20}$.
\label{fig:fact2}
}
\end{center}
\end{figure}
% ---------------------------------------------------------------------
\par
The target matrix in the second example is a $5$-by-$5$ Hessenberg matrix:
\[
A^{(0)} = 
\left( \begin{array}{ccccc}
1 & 2 & 1 & 3 & 2 \\
1 & 4 & 2 & 6 & 4 \\
 & 4 & 3 & 9 & 6 \\
 & & 3 & 12 & 8 \\
 & & & 12 & 10
\end{array} \right).
\]
Since $A^{(0)}$ can be decomposed by TN matrices as:
\[
A^{(0)}=
\left( \begin{array}{ccccc}
1 & & & & \\
1 & 1 & & & \\
 & 2 & 1 & & \\
 & & 3 & 1 & \\
 & & & 4 & 1
\end{array} \right)
\left( \begin{array}{ccccc}
1 & 1 & 1 & 1 & 1 \\
 & 1 & 1 & 1 & 1 \\
 & & 1 & 1 & 1 \\
 & & & 1 & 1 \\
 & & & & 1
\end{array} \right)
\left( \begin{array}{ccccc}
1 & \\
 & 2 \\
 & & 1 \\
 & & & 3 \\
 & & & & 2
\end{array} \right), 
\]
$A^{(0)}$ is also TN. The eigenvalues computed by the function {\tt Eigenvalues} are $\lambda_{1}=22.4186804701347$, $\lambda_{2}=5.58970261546314, \lambda_{3}=1.39103188993094, \lambda_{4}=0.446357128198326$, and $\lambda_{5}=0.154227896272925$. In this case, we cannot apply the dhToda equation directly to compute the eigenvalues because the bidiagonal decomposition of the initial matrix is not given. Furthermore, the dhToda equation requires that the lower subdiagonal entries to be all $1$. 
Similarly to the first example, without the bidiagonal factorization of $A^{(0)}$, we can set the initial values and parameters as 
$x_{1,1}^{(0)}=1$, $x_{1,2}^{(0)}=2$, $x_{1,3}^{(0)}=1$, $x_{1,4}^{(0)}=3$, $x_{1,5}^{(0)}=2$, 
$x_{2,2}^{(0)}=4$, $x_{2,3}^{(0)}=2$, $x_{2,4}^{(0)}=6$, $x_{2,5}^{(0)}=4$, 
$x_{3,3}^{(0)}=3$, $x_{3,4}^{(0)}=9$, $x_{3,4}^{(0)}=6$, 
$x_{4,4}^{(0)}=12$, $x_{4,5}^{(0)}=8$, $x_{5,5}^{(0)}=10$, 
$y_{1}^{(0)}=1$, $y_{2}^{(0)}=4$, $y_{3}^{(0)}=3$, $y_{4}^{(0)}=12$, 
$m=5$, and $M=5$. 
We also adopt $\mu^{(n)}=1$ which is the same as in the first example. 
Applying a discrete-time evolution from $n=0$ to $n=150$, we obtain 
$\vert x_{1,1}^{(150)}-\lambda_{1}\vert/\lambda_{1}=1.8896858079126828\times10^{-17}, 
\vert x_{2,2}^{(150)}-\lambda_{2}\vert/\lambda_{2}=7.321880894841474\times10^{-16}, 
\vert x_{3,3}^{(150)}-\lambda_{3}\vert/\lambda_{3}=1.7531315204931288\times10^{-16},
\vert x_{4,4}^{(150)}-\lambda_{4}\vert/\lambda_{4}=9.700460735364932\times10^{-15}$, and 
$\vert x_{5,5}^{(150)}-\lambda_{5}\vert/\lambda_{5}=2.465516259530728\times10^{-14}$. 
Even in the case where the dhToda equation is not applicable, we thus see that the extended $q$-discrete Toda equation \eqref{eqn:extended_qdToda} generates good approximations of the eigenvalues.
%
%
% ---------------------------------------------------------------------
\section{Concluding remarks}
\label{sec:5}
% ---------------------------------------------------------------------
%
In this paper, we proposed an extension of the $q$-discrete Toda equation and related it to similarity transformations of Hessenberg matrices. Next, by introducing implicit $L$ theorem, we showed its relationship to the sifted $LR$ transformations of Hessenberg matrices. By assuming that the initial matrix is totally nonnegative (TN), we clarified the convergence of the extended $q$-discrete Toda variables to eigenvalues of Hessenberg matrices. Finally, we presented two numerical examples that demonstrate the extended $q$-discrete Toda equation's convergence to TN Hessenberg eigenvalues in the extended $q$-discrete Toda equation. Remarkably, in contrast to applications of discrete hungry integrable systems, bidiagonal factorizations are not necessary for the extended $q$-discrete Toda equation to used for computing eigenvalues. 
% ---------------------------------------------------------------------
\par
In future work, we plan to investigate the numerical stability and convergence rate of eigenvalue computation by the extended $q$-discrete Toda equation. Like the quotient-difference algorithm, we will simultaneously attempt to introduce auxiliary variables to avoid cancellation. Since the extended $q$-discrete Toda equation is related to the implicit-shift $LR$ transformations, we also aim to design a shift strategy to accelerate the convergence. 
\section*{Acknowledgements}
%\textcolor{red}{The authors thank the reviewer for his/her careful reading and constructive suggestions.}
This work was partially supported by the joint project of Kyoto University and Toyota Motor Corporation, titled ``Advanced Mathematical Science for Mobility Society''.
\section*{ORCID}
R. Watanabe: 0000-0002-5758-9587
%
%
% ---------------------------------------------------------------------

%%
%

\begin{thebibliography}{50}
%
\bibitem{Area_2018}I. Area, A. Branquinho and A.F. Moreno and E. Godoy, 
\emph{Orthogonal polynomial interpretation of $q$-Toda and $q$-Volterra equations}, 
Bull. Malays. Math. Sci. Soc. 41 (2018), pp. 393--414.
%
\bibitem{Bogoyavlenskii_1991}O.I. Bogoyavlenskii, 
\emph{Algebraic constructions of integrable dynamical systems extensions of the Volterra system}, 
Russian Math. Surveys 46 (1991), pp. 1--64.
%
\bibitem{Chu_1986}M.T. Chu, 
\emph{A differential equation approach to the singular value decomposition of bidiagonal matrices}, 
Linear Algebra Appl. 80 (1986), pp. 71--79. 
%
\bibitem{Fallat_2005}S.M. Fallat and M. I. Gekhtman, \emph{Jordan structures of totally nonnegative matrices}, Canad. J. Math. 57 (2005), pp. 82--98.
%
\bibitem{Fukuda_2013_1}A. Fukuda, E. Ishiwata, Y. Yamamoto, M. Iwasaki and Y. Nakamura, 
\emph{Integrable discrete hungry systems and their related matrix eigenvalues}, 
Annal. Mat. Pura Appl. 192 (2013), pp. 423--445. 
%
\bibitem{Fukuda_2013_2}A. Fukuda, Y. Yamamoto, M. Iwasaki, E. Ishiwata and Y. Nakamura, 
\emph{On a shifted $LR$ transformation derived from the discrete hungry Toda equation}, 
Monatsh. Math. 170 (2013), pp. 11--26. 
%
\bibitem{Gaska_2010}M. Gaska and J. M. Pe\~{n}a, 
\emph{On factorizations of totally positive matrices}, in {\it Total Positivity and Its
Applications} (M. Gasca and C. A. Micchelli (eds.)), Kluwer Academic Publishers, 2010.
%
\bibitem{Golub_1996}G.H. Golub and C.F. Van Loan, {\it Matrix Computations, 3rd ed.}, Baltimore, MD: Johns Hopkins Univ. Press, 1996. 
%
\bibitem{Hirota_1981}R. Hirota, 
\emph{Discrete analogue of a generalized Toda equation}, 
J. Phys. Soc. Jpn. 50 (1981), pp. 3785--3791. 
%
\bibitem{Itoh_1987}Y. Itoh, 
\emph{Integrals of a Lotka-Volterra system of odd number of variables}, 
Prog. Theor. Phys. 78 (1987), pp. 507--510. 
%
\bibitem{Iwasaki_2002}M. Iwasaki and Y. Nakamura, 
\emph{On the convergence of a solution of the discrete Lotka-Volterra system}, 
Inverse Probl. 18 (2002), pp. 1569--1578. 
%
\bibitem{Iwasaki_2004}M. Iwasaki and Y. Nakamura, 
\emph{An application of the discrete Lotka-Volterra system with variable step-size to singular value computation}, 
Inverse Probl. 20 (2004), pp. 553--563. 
%
\bibitem{Kostant_1979}B. Kostant, 
\emph{The solution to a generalized Toda lattice and representation theory}, 
Adv. Math. 34 (1979), pp. 195--338. 
%
\bibitem{Li_2002}C.-K. Li and R. Mathias, 
\emph{Interlacing inequalities for totally nonnegative matrices}, 
Linear Algebra Appl. 341 (2002), pp. 35--44.
%
\bibitem{Pinkus_2009}A. Pinkus, 
{\it Totally Positive Matrices (Cambridge Tracts in Mathematics)}, Cambridge: Cambridge University Press, 2009. 
%
\bibitem{Rolania_2019}D.B. Rolan\`ia, 
\emph{On the Darboux transform and the solutions of some integrable systems}, 
Rev. R. Acad. Cienc. Exactas F\`is. Nat. Ser. A Mat. RACSAM 113 (2019), pp. 1359--1378.
%
\bibitem{Rutishauser_1990}H. Rutishauser, 
{\it Lectures on Numerical Mathematics}, 
Birkh\"auser, Boston, 1990. 
%
\bibitem{Shinjo_2020}M. Shinjo, M. Iwasaki and K. Kondo, 
\emph{The Kostant-Toda equation and the hungry integrable systems}, 
J. Math. Anal. Appl. 483 (2020), 123627(15pp). 
%
\bibitem{Suris_1996}Y.B. Suris, 
\emph{Integrable discretizations of the Bogoyavlensky lattices}, 
J. Math. Phys. 37 (1996), pp. 3982--3996. 
%
\bibitem{Symes_1982}W.W. Symes, 
\emph{The QR algorithm and scattering for the finite nonperiodic Toda lattice}, 
Physica 4 (1982), pp. 275--280. 
%
\bibitem{Toda_1967}M. Toda, 
\emph{Vibration of a chain with nonlinear integration}, 
J. Phys. Soc. Jpn. 22 (1967), pp. 431--436. 
%
\bibitem{Tokihiro_1999}T. Tokihiro, A. Nagai and J. Satsuma, 
\emph{Proof of solitonical nature of box and ball systems by means of inverse ultra-discretization}, 
Inverse Probl. 15 (1999), pp. 1639--1662. 
%
\bibitem{Tsujimoto_1993}S. Tsujimoto,  R. Hirota and S. Oishi, 
\emph{An extension and discretization of Volterra equation I}, 
Tech. Rep. Proc. IEICE NLP 92 (1993), pp. 1--3. 
%
\bibitem{Watanabe_2022}R. Watanabe, M. Shinjo and M. Iwasaki, 
\emph{Matrix similarity transformations derived from extended $q$-analogues of Toda equation and Lotka-Volterra system}, 
J. Differ. Equ. Appl. (under review). 
%
\end{thebibliography}
\end{document}